\documentclass[12pt]{amsart}
\usepackage{amsmath}
\usepackage{amsfonts}
\usepackage{amssymb}
\usepackage{subfigure}
\usepackage{graphicx}
\usepackage{bookmark}

\numberwithin{figure}{section}
\newtheorem{lemma}{Lemma}[section]

\newtheorem{definition}[lemma]{Definition}

\newtheorem{example}[lemma]{Example}
\newtheorem{proposition}[lemma]{Proposition}
\newtheorem{remark}[lemma]{Remark}
\newtheorem{theorem}[lemma]{Theorem}
\numberwithin{equation}{section}

\begin{document}
\title[Entanglement entropy and hyperuniformity]{Entanglement entropy and
hyperuniformity of Ginibre and Weyl-Heisenberg ensembles }
\author{Lu\'{\i}s Daniel Abreu}
\address{NuHAG, Faculty of Mathematics, University of Vienna,
Oskar-Morgenstern-Platz 1, A-1090, Vienna, Austria}
\email{abreuluisdaniel@gmail.com}
\thanks{The author was supported by the Austrian Science Fund (FWF) via the
project (P31225-N32). All data generated or analysed during this study are
included in this published article. On behalf of all authors, the
corresponding author states that there is no conflict of interest.}

\begin{abstract}
We show that, for a class of planar determinantal point processes (DPP) $%
\mathcal{X}$, the growth of the entanglement entropy $S(\mathcal{X}(\Omega ))
$ of $\mathcal{X}$\ on a compact region $\Omega \subset \mathbb{R}^{2d}$, is
related to the variance $\mathbb{V}\left( \mathcal{X}(\Omega )\right) $ as
follows:%
\begin{equation*}
\mathbb{V}\left( \mathcal{X}(\Omega )\right) \lesssim S\left( \mathcal{X(}%
\Omega \mathcal{)}\right) \lesssim \mathbb{V}\left( \mathcal{X}(\Omega
)\right) \text{.}
\end{equation*}%
Therefore, such DPPs satisfy an \emph{area law} $S(\mathcal{X(}\Omega 
\mathcal{))}\lesssim \left\vert \partial {\Omega }\right\vert $, where $%
\partial {\Omega }$ is the boundary of $\Omega $) if they are of \emph{Class
I hyperuniformity }($\mathbb{V}\left( \mathcal{X}(\Omega )\right) \lesssim
\left\vert \partial {\Omega }\right\vert $), while the \emph{area law is
violated} if they are of \emph{Class II hyperuniformity }(as \ $L\rightarrow
\infty $, $\mathbb{V}\left( \mathcal{X}(L\Omega )\right) \sim C_{\Omega
}L^{d-1}\log L$). As a result, the entanglement entropy of Weyl-Heisenberg
ensembles (a family of DPPs containing the Ginibre ensemble and Ginibre-type
ensembles in higher Landau levels), satisfies an area law, as a consequence
of its hyperuniformity.
\end{abstract}

\maketitle

\section{Introduction}

If one considers a partition of a many-particle state in two subregions, 
\emph{the entanglement entropy} measures the degree of entanglement between
the two regions, which is given by the von Neumann entropy of the reduced
state in one of the regions. Entanglement entropy is nowadays a widely
studied quantity in many-particle interacting systems \cite%
{EntaglementEntropy,Area,WidomConj,Ginibre,FL,NC}. In this note we interpret
the definition of entanglement entropy for fermionic states given in \cite[%
Proposition 7.2]{EntaglementEntropy}, in terms of planar determinantal point
process (DPP) in $\mathbb{R}^{2d}$. This allows to define the entanglement
entropy $S(\mathcal{X}(\Omega ))$ of a DPP $\mathcal{X}$\ in a compact
subregion $\Omega \subset \mathbb{R}^{2d}$, as $S(\mathcal{X}(\Omega ))=%
\mathrm{trace}(f(T_{\Omega }))$, where $f(x)=-x\ln x-(1-x)\ln (1-x)$ and $%
T_{\Omega }$ is a Toeplitz operator defined with the correlation kernel of
the DPP, with symbol the indicator function of $\Omega $ (see Section 2).
Motivated by this definition, we will show that, for a class of planar DPPs
for which $\mathrm{trace}\left( T_{\Omega }^{p}\left( 1-T_{\Omega }\right)
^{p}\right) $ is bounded for $0<p<1$, which include the Ginibre ensemble and
its higher Landau level versions \cite{Shirai}, the following relations
between the entanglement entropy $S(\mathcal{X}(\Omega ))$ and the variance $%
\mathbb{V}\left( \mathcal{X}(\Omega )\right) $ hold:%
\begin{equation}
\mathbb{V}\left( \mathcal{X}(\Omega )\right) \lesssim S\left( \mathcal{X(}%
\Omega \mathcal{)}\right) \lesssim \mathbb{V}\left( \mathcal{X}(\Omega
)\right) \text{.}  \label{ine}
\end{equation}%
The only related inequality we found in the literature is the one in \cite[%
(6)]{WidomConj}, which holds with no assumptions, but has a logarithm
correction term on the upper bound. The entanglement entropy is said to
satisfy an \emph{area law} if $S\left( \mathcal{X(}\Omega \mathcal{)}\right)
\lesssim \left\vert \partial {\Omega }\right\vert $, where $\left\vert
\partial {\Omega }\right\vert $ is the measure of the perimeter of $\Omega $
or, asymptotically, for a dilated region $R{\Omega }$, if $\ \mathbb{V}%
\left( \mathcal{X}(\Omega \mathcal{)}\right) \sim R^{d-1}$ as $R\rightarrow
\infty $. Our results imply an area law $S\left( \mathcal{X(}\Omega \mathcal{%
)}\right) \lesssim \left\vert \partial {\Omega }\right\vert $ when $\mathcal{%
X}$ is the \emph{infinite Ginibre ensemble }with kernel given as%
\begin{equation*}
K_{0}(z,w)=e^{-\frac{\pi }{2}(\left\vert z\right\vert ^{2}+\left\vert
w\right\vert ^{2})}e^{\pi \overline{z}w}
\end{equation*}%
and when $\mathcal{X}$ is one of the Ginibre-type ensembles \cite{Shirai},
defined with the reproducing kernel of the $n$ eigenspace of the Landau
operator $L_{z}:=-\partial _{z}\partial _{\overline{z}}+\pi \overline{z}%
\partial _{\overline{z}}$,%
\begin{equation*}
K_{n}(z,w)=e^{-\frac{\pi }{2}(\left\vert z\right\vert ^{2}+\left\vert
w\right\vert ^{2})}L_{n}(\pi \left\vert z-w\right\vert ^{2})e^{\pi \overline{%
z}w}\text{.}
\end{equation*}%
Our main result will be stated for \emph{The Weyl-Heisenberg ensemble} $%
\mathcal{X}_{g}$ on $\mathbb{R}^{2d}$ introduced in\emph{\ }\cite{APRT} and
studied further in \cite{abgrro17,Partial}, a family of DPPs depending on a
window function $g\in L^{2}(\mathbb{R}^{d})$, with correlation kernel%
\begin{equation*}
K_{g}(z,w)={K}_{g}((x,\xi ),(x^{\prime },\xi ^{\prime }))=\int_{\mathbb{R}%
^{d}}e^{2\pi i(\xi ^{\prime }-\xi )t}g(t-x^{\prime })\overline{g(t-x)}dt%
\text{.}
\end{equation*}%
When $g$ is a Gaussian, $K_{g}(z,w)$ becomes a weighted version of $%
K_{0}(z,w)$ and when $g$ is a Hermite function, it becomes a weighted
version of $K_{n}(z,w)$. More details about these specializations will be
given in section 3.

A DPP $\mathcal{X}$ is said to be \emph{hyperuniform of Class I} \cite%
{TorStil},\cite[(97) and Table 1]{HyperSurvey}, if $\ \mathbb{V}\left( 
\mathcal{X}(\Omega \mathcal{)}\right) \lesssim \left\vert \partial {\Omega }%
\right\vert $ or, asymptotically, for a dilated region $R{\Omega }$, if $\ 
\mathbb{V}\left( \mathcal{X}(\Omega \mathcal{)}\right) \sim R^{d-1}$ as $%
R\rightarrow \infty $. As a result of (\ref{ine}), area laws for the DPPs
considered in this paper will follow as a consequence of their Class I
hyperuniformity of rate $1$. Hyperuniform states of matter are correlated
systems characterized by the suppression of density fluctuations at large
scales \cite{HyperSurvey,TorStil,TorFermionic,GL,GL1,APRT}. While the
relation (\ref{ine}) suggests what seems to be a hitherto unnoticed relation
between the concepts of entanglement entropy and of hyperuniformity,
similarities between the entanglement entropy and variance fluctuations have
been empirically observed in several contexts \cite{NC}, suggesting that
both concepts may be used to quantify the level of supression of
fluctuations at large scales typical of a number of physical and
mathematical systems known as hyperuniform \cite[(97) and Table 1]%
{HyperSurvey}. The inequality (\ref{ine}) is a first step towards a
mathematical proof of this hypothesis.

The presentation of this note is organized as follows. The next section
contains the concepts of entanglement entropy and number variance for DPPs
and proves the inequality (\ref{ine}) under the assumptions on $\mathrm{trace%
}\left( T_{\Omega }^{p}\left( 1-T_{\Omega }\right) ^{p}\right) $. The third
section introduces some notions about the Weyl-Heisenberg ensemble, and (\ref%
{ine}) is assured to hold for this case, thanks to the bounds of $\mathrm{%
trace}\left( T_{\Omega }^{p}\left( 1-T_{\Omega }\right) ^{p}\right) $,\
recently obtained by Marceca and Romero \cite{SD}. We then state and prove
the bound $S\left( \mathcal{X}_{g}\mathcal{(}\Omega \mathcal{)}\right)
\lesssim \left\vert \partial {\Omega }\right\vert $ on the entanglement
entropy of Weyl-Heisenberg ensembles. A lower bound $\left\vert \partial {%
\Omega }\right\vert \lesssim S\left( \mathcal{X}_{g}\mathcal{(}\Omega 
\mathcal{)}\right) $ is also observed to hold under some extra assumptions,
and the important examples of Ginibre and of Shirai's Ginibre-type ensembles
on higher Landau levels \cite{Shirai} are used to illustrate the scope of
the result on Weyl-Heisenberg ensembles. In the last section, bounds on the
entropy using the construction of finite Weyl-Heisenberg ensembles \cite%
{abgrro17} are obtained.

\section{Entanglement entropy and variance of DPPs}

We refer to \cite{DetPointRand,Partial} for precise definitions and
background on DPPs. A locally integrable kernel $K(z,w)$ defines the
correlation kernel of a determinantal point process (DPP) distributing $%
\mathcal{X}\left( \Omega \right) $ points in $\Omega \subset \mathbb{R}^{2d}$%
, whose $k$-point intensities are given by $\rho _{k}(z_{1},...,z_{k})=\det
\left( K(z_{i},z_{j})\right) _{1\leq i,j\leq k}$. The $1$-point intensity of 
$\mathcal{X}$ is then given by $\rho _{1}(z)=K(z,z)$, allowing to compute
the expected number of points that fall in $\Omega $ as%
\begin{equation*}
\mathbb{E}\left[ \mathcal{X(}\Omega \mathcal{)}\right] \mathbb{=}%
\int_{\Omega }K(z,z)dz\text{,}
\end{equation*}%
while the number variance in $\Omega $ is given as (see \cite[pg. 40]%
{Variance} for a detailed proof): 
\begin{equation}
\mathbb{V}\left( \mathcal{X(}\Omega \mathcal{)}\right) =\mathbb{E}\left[ 
\mathcal{X(}\Omega \mathcal{)}^{2}\right] -\mathbb{E}\left[ \mathcal{X}_{g}%
\mathcal{(}\Omega \mathcal{)}\right] ^{2}=\int_{\Omega }K\left( z,z\right)
dz-\int_{\Omega ^{2}}\left\vert K\left( z,w\right) \right\vert ^{2}dzdw\text{%
.}  \label{var}
\end{equation}

Consider a compact set $\Omega \subset \mathbb{R}^{2d}$. The entanglement
entropy $S(\mathcal{X}(\Omega ))$ measures the degree of entanglement of the
DPP $\mathcal{X}$ reduced to the region $\Omega $. A DPP satisfies \emph{an
area law} if the leading term of the entanglement entropy grows at most
proportionally with the measure of the boundary of the partition defining
the reduced state \cite{Area,EntaglementEntropy}. In $\mathbb{R}^{2d}=\Omega
\cup \Omega ^{c}$ this corresponds to a growth of the order of the perimeter 
$\left\vert \delta \Omega \right\vert $. The set $\Omega \subseteq {\mathbb{R%
}^{2d}}$ is said to have \emph{finite perimeter} if its characteristic
function $1_{\Omega }$ is of bounded variation (the concept of `area law'
for the entanglement entropy would be, with this terminology, more precisely
named as `perimeter law', but we keep up with the traditional terminology).
In this case, its perimeter is $\left\vert \partial \Omega \right\vert :=%
\mathit{Var}(1_{\Omega })$.

Our analysis is based on associating to the kernel of $\mathcal{X}$, $K(z,w)$
(a locally integrable reproducing kernel of a Hilbert space $H\subset
L^{2}\left( \mathbb{R}^{2d}\right) $), the following operator:%
\begin{equation*}
(T_{\Omega }f)(z)=\int_{\Omega }f(w)\overline{K(z,w)}dw\text{,}
\end{equation*}%
where $dw$ stands for Lebesgue measure, mapping $f$ to a smooth function in $%
L^{2}\left( \mathbb{R}^{2d}\right) $ with most of its energy concentrated in
the region $\Omega $. Since $\Omega \subset \mathbb{R}^{2d}$\ is compact and 
$K(z,w)$ locally integrable, $T_{\Omega }$ is a compact positive
(self-adjoint) operator of trace class, and one can invoke the spectral
theorem to assure that $T_{\Omega }$ is diagonalized by an orthonormal set
of eigenfunctions $\{e_{n}^{\Omega }(z):n\geq 1\}$ with corresponding
eigenvalues $\{\lambda _{n}^{\Omega }:n\geq 1\}$ ordered non-increasingly.
The operator is positive and bounded by 1 (see \cite[Lemma 2.1]{AGR} for
details in the Weyl-Heisenberg case).

For the definition of entanglement entropy of a DPP on a region $\Omega $ we
will use the result in Proposition 7.2 of \cite{EntaglementEntropy}.

\begin{definition}
The entanglement entropy $S(\mathcal{X}(\Omega ))$ of the DPP $\mathcal{X}$
on a compact set $\Omega \subset \mathbb{R}^{2d}$ is defined in terms of $%
T_{\Omega }$ as%
\begin{equation*}
S(\mathcal{X}(\Omega ))=\mathrm{trace}(f(T_{\Omega }))\text{,}
\end{equation*}%
where%
\begin{equation}
f(x)=-x\ln x-(1-x)\ln (1-x)\text{.}  \label{f}
\end{equation}
\end{definition}

The traces of $T_{\Omega }$ and $T_{\Omega }^{2}$ are given by ($K\left(
z,z\right) =1$) 
\begin{align}
& \mathrm{trace}(T_{\Omega })=\int_{\Omega }K(z,z)dz=\mathbb{E}\left[ 
\mathcal{X(}\Omega \mathcal{)}\right] ={|\Omega |}=\sum_{n\geq 1}{\lambda
_{n}^{\Omega }},  \label{eq_trace1} \\
& \mathrm{trace}(T_{\Omega }^{2})=\int_{\Omega ^{2}}\left\vert K\left(
z,w\right) \right\vert ^{2}dzdw=\sum_{n\geq 1}{(\lambda _{n}^{\Omega }})^{2}%
\text{.}  \label{eq_trace2}
\end{align}%
and the number variance of $\mathcal{X}(\Omega )$, according to (\ref{var}),
by 
\begin{equation}
\mathbb{V}\left( \mathcal{X}(\Omega )\right) =\mathrm{trace}(T_{\Omega })-%
\mathrm{trace}(T_{\Omega }^{2})=\sum_{n\geq 1}{\lambda _{n}^{\Omega }-}%
\sum_{n\geq 1}{(\lambda _{n}^{\Omega }})^{2}\text{.}  \label{varrel}
\end{equation}

It has been drawn to the attention of the author by Gr\"{o}chenig \cite%
{Groch} that, for $x\in \left[ 0,1\right] $, the following inequality can be
easily proved:%
\begin{equation}
4x(x-1)\leq \frac{1}{\log 2}f(x)\text{.}  \label{lowerpointwise}
\end{equation}%
where $f(x)=-x\ln x-(1-x)\ln (1-x)$, so that $\mathbb{V}\left( \mathcal{X}%
(\Omega )\right) =\mathrm{trace}\left( T_{\Omega }-T_{\Omega }^{2}\right)
\leq \frac{1}{4\log 2}\mathrm{trace}(f(T_{\Omega }))$. This leads to a lower
bound for the entanglement entropy%
\begin{equation}
\mathbb{V}\left( \mathcal{X}(\Omega )\right) \lesssim S\left( \mathcal{X(}%
\Omega \mathcal{)}\right) \text{.}  \label{lower}
\end{equation}

Inequality (\ref{lower}) has been used before to show the violation of the
area law by fermionic process (see \cite{WidomConj} and the references
therein, where also upper inequalities for the entropy in terms of the
variance, with a log correction term, are obtained). For $x\in \left[ 0,1%
\right] $ one cannot expect a pointwise upper bound for $f(x)=-x\ln
x-(1-x)\ln (1-x)$ as a constant times $x(x-1)$, due to the singularities of $%
f(x)$. Nevertheless, under a boundedness conditions on the so-called
Schatten $p$-norms of $T_{\Omega }-T_{\Omega }^{2}$, it is possible to prove
an upper bound by relating the trace of the functions of positive
self-adjoint operators bounded by 1.

Our main results will depend on the following inequality, conditioned to a
bound on the Schatten $p$-norms of $T_{\Omega }-T_{\Omega }^{2}$.

\begin{proposition}
Let $\mathcal{X}$ be a DPP on $\mathbb{R}^{2d}$ such that the associated
operator $T_{\Omega }$ is self-adjoint, positive, bounded by 1 and is of
trace class satisfying, for $0<p<1$, 
\begin{equation}
\mathrm{trace}\left( T_{\Omega }^{p}\left( 1-T_{\Omega }\right) ^{p}\right)
\leq C\text{,}  \label{Shatten}
\end{equation}%
where $C$ depends on $\Omega $ and $p$. Then the entanglement entropy and
the variance of $\mathcal{X(}\Omega \mathcal{)}$ satisfy%
\begin{equation}
\mathbb{V}\left( \mathcal{X}(\Omega )\right) \lesssim S\left( \mathcal{X(}%
\Omega \mathcal{)}\right) \lesssim \mathbb{V}\left( \mathcal{X}(\Omega
)\right) \text{.}  \label{ineEE}
\end{equation}
\end{proposition}

\begin{proof}
Observe that $f(x)=-x\ln x-(1-x)\ln (1-x)$ belongs to the class of
continuous function such that $\left\vert f(t)\right\vert =O(t^{p})$\ and $%
\left\vert f(1-t)\right\vert =O(t^{p})$ as $t\rightarrow 0$ with $p>0$.
Strongly inspired by the idea of \cite[Theorem 6.2]{EntaglementEntropy}, we
will prove that, for $f$ in this class, if $\Omega \subset \mathbb{R}^{2d}$
is compact, then%
\begin{equation*}
\mathrm{trace}(f(T_{\Omega }))\lesssim \mathbb{V}\left( \mathcal{X}(\Omega
)\right) \text{.}
\end{equation*}%
The proof will use that $\mathrm{trace}$ is a positive linear functional, in
the sense that if $f\leq g$ then $\mathrm{trace}\left( f\right) \leq \mathrm{%
trace}\left( g\right) $, and relate $\mathrm{trace}(f(T_{\Omega }))$ to $%
\mathbb{V}\left( \mathcal{X}(\Omega )\right) $ using the identity (\ref%
{varrel}), first for polynomials vanishing at $0$ and $1$ and then for
functions of the form $f(z)=g(x)h_{p}(x)$ with $h_{p}(x)=x^{p}(1-x)^{p}$\
and $g\in C(\left[ 0,1\right] )$ such that $g(0)=g(1)=0$, using polynomial
approximation.

\textbf{Step 1.} In this step we prove that $\mathrm{trace}(P_{n}(T_{\Omega
}))\lesssim \mathbb{V}\left( \mathcal{X}(\Omega )\right) $, where $P_{n}$ is
a polynomial of degree $n$, such that $P(0)=P(1)=0$. For $k\geq 1$,%
\begin{equation*}
\mathrm{trace}\left( T_{\Omega }^{k}\right) -\mathrm{trace}\left( T_{\Omega
}^{k+1}\right) =\mathrm{trace}\left( T_{\Omega }^{k-1}\left( T_{\Omega
}-T_{\Omega }^{2}\right) \right) \text{.}
\end{equation*}%
Since $T_{\Omega }-T_{\Omega }^{2}$ is\ a non-negative defined operator, we
can use the inequality $\mathrm{trace}\left( AB\right) \leq \left\Vert
A\right\Vert \mathrm{trace}$\textrm{(}$B$\textrm{), }together with $%
\left\Vert T_{\Omega }^{k-1}\right\Vert \leq 1$, to obtain 
\begin{equation*}
\mathrm{trace}\left( T_{\Omega }^{k}\right) -\mathrm{trace}\left( T_{\Omega
}^{k+1}\right) \leq \mathrm{trace}\left( T_{\Omega }-T_{\Omega }^{2}\right) =%
\mathbb{V}\left( \mathcal{X}(\Omega )\right) \text{.}
\end{equation*}%
Since a general polynomial vanishing at $0$ and $1$ can be written as linear
combinations of $x^{k}-x^{k+1}$, we write%
\begin{equation*}
P_{n}(x)=\sum_{k=0}^{n}a_{k}\left( x^{k}-x^{k+1}\right)
\end{equation*}
and the above gives, by linearity, 
\begin{equation*}
\mathrm{trace}\left( P_{n}(T_{\Omega })\right) =\sum_{k=0}^{n}a_{k}\left( 
\mathrm{trace}\left( T_{\Omega }^{k}\right) -\mathrm{trace}\left( T_{\Omega
}^{k+1}\right) \right) \lesssim \mathbb{V}\left( \mathcal{X}(\Omega )\right) 
\text{.}
\end{equation*}%
\textbf{Step 2. }We show that, for every $p>0$, $\mathrm{trace}\left(
h_{p}\left( T_{\Omega }\right) \right) $ is bounded, where $%
h_{p}(x)=x^{p}(1-x)^{p}$, $0<x<1$. For $p\geq 1$ and $0<x<1$, we have $%
x^{p}(1-x)^{p}\leq x(1-x)$ and%
\begin{equation*}
\mathrm{trace}\left( h_{p}\left( T_{\Omega }\right) \right) =\mathrm{trace}%
\left( T_{\Omega }^{p}\left( 1-T_{\Omega }\right) ^{p}\right) \leq \mathrm{%
trace}\left( T_{\Omega }-T_{\Omega }^{2}\right) =\mathbb{V}\left( \mathcal{X}%
(\Omega )\right) \text{.}
\end{equation*}%
For $0<p<1$ and $0<x<1$, it follows from the hypothesis (\ref{Shatten}) that 
$\mathrm{trace}\left( h_{p}\left( T_{\Omega }\right) \right) $ is bounded by 
$C>0$. We have thus%
\begin{equation*}
\mathrm{trace}\left( h_{p}\left( T_{\Omega }\right) \right) \leq C_{0}\text{,%
}
\end{equation*}%
where $C_{0}=\max \{\mathbb{V}\left( \mathcal{X}(\Omega )\right) ,C\}$.

\textbf{Step 3. }For the extension to continuous functions $f$ such that $%
\left\vert f(t)\right\vert =O(t^{p})$\ and $\left\vert f(1-t)\right\vert
=O(t^{p})$ as $t\rightarrow 0$ with $p>0$, we use a polynomial approximation
argument as in \cite[Theorem 6.2]{EntaglementEntropy}. For a $p>0$ one can
write $f$ as $f(z)=g(x)h_{p}(x)$ with $h_{p}(x)=x^{p}(1-x)^{p}$\ and $g\in C(%
\left[ 0,1\right] )$ such that $g(0)=g(1)=0$. Given $\epsilon >0$ we can
invoke the Weierstrass approximation theorem to find a polynomial $P(x)$
such that $P(0)=P(1)=0$ and $\left\vert g-P\right\vert <\epsilon $. Thus, $%
\mathrm{trace}\left( f\left( T_{\Omega }\right) \right) =\mathrm{trace}%
\left( gh_{p}\left( T_{\Omega }\right) \right) $ and the polynomial
approximation of $g$ by $P$ allows one to write $g\leq P+\epsilon $ and 
\begin{equation}
\mathrm{trace}\left( f\left( T_{\Omega }\right) \right) =\mathrm{trace}%
\left( gh_{p}\left( T_{\Omega }\right) \right) \leq \mathrm{trace}\left(
Ph_{p}\left( T_{\Omega }\right) \right) +\epsilon \mathrm{trace}\left(
h_{p}\left( T_{\Omega }\right) \right) \text{.}  \label{Weier}
\end{equation}

Combining with Step 2, we arrive at 
\begin{equation}
\mathrm{trace}\left( f\left( T_{\Omega }\right) \right) \leq \mathrm{trace}%
\left( Ph_{p}\left( T_{\Omega }\right) \right) +\epsilon C_{0}\text{.}
\label{est1}
\end{equation}%
Since $P(0)=P(1)=0$, there exists a polynomial $P_{1}(x)$ such that $%
P(x)=P_{1}(x)h_{1}(x)$, leading to $P(x)h_{p}(x)=P_{1}(x)h_{p}(x)h_{1}(x)$.
This allows to control $\mathrm{trace}\left( Ph_{p}\left( T_{\Omega }\right)
\right) $, by writing $g(x)=P_{1}(x)h_{p}(x)$ and invoking Weierstrass
approximation of $g(x)$\ by another polynomial $P_{2}(x)$. For an $\epsilon
_{1}>0$ we obtain, since $g\leq P_{2}+\epsilon _{1}$,%
\begin{equation}
\mathrm{trace}\left( Ph_{p}\left( T_{\Omega }\right) \right) =\mathrm{trace}%
\left( gh_{1}\left( T_{\Omega }\right) \right) \leq \mathrm{trace}\left(
P_{2}h_{1}\left( T_{\Omega }\right) \right) +\epsilon _{1}\mathrm{trace}%
\left( h_{1}\left( T_{\Omega }\right) \right) \text{.}  \label{est2}
\end{equation}%
By Step 1, since $P_{2}(x)h_{1}(x)$ is a polynomial, 
\begin{equation*}
\mathrm{trace}\left( P_{2}h_{1}\left( T_{\Omega }\right) \right) \lesssim 
\mathbb{V}\left( \mathcal{X}(\Omega )\right) \text{.}
\end{equation*}%
Observing that%
\begin{equation*}
\mathrm{trace}\left( h_{1}\left( T_{\Omega }\right) \right) =\mathrm{trace}%
\left( T_{\Omega }-T_{\Omega }^{2}\right) =\mathbb{V}\left( \mathcal{X}%
(\Omega )\right) \text{,}
\end{equation*}%
then (\ref{est2}) leads to%
\begin{equation*}
\mathrm{trace}\left( Ph_{p}\left( T_{\Omega }\right) \right) \lesssim 
\mathbb{V}\left( \mathcal{X}(\Omega )\right) +\epsilon _{1}\mathbb{V}\left( 
\mathcal{X}(\Omega )\right) \text{.}
\end{equation*}%
It follows from (\ref{est1}) that%
\begin{equation*}
\mathrm{trace}\left( f\left( T_{\Omega }\right) \right) \lesssim \mathbb{V}%
\left( \mathcal{X}(\Omega )\right) +\epsilon _{1}\mathbb{V}\left( \mathcal{X}%
(\Omega )\right) +\epsilon C_{0}\text{.}
\end{equation*}%
Since $\epsilon $ and $\epsilon _{1}$ are at our disposal, this implies $%
\mathrm{trace}\left( f\left( T_{\Omega }\right) \right) \lesssim \mathbb{V}%
\left( \mathcal{X}(\Omega )\right) $.
\end{proof}

\section{Entanglement entropy of Weyl-Heisenberg ensembles}

The main result will be stated in terms of \emph{Weyl-Heisenberg ensembles}.
This includes as special cases the Ginibre ensemble and its higher Landau
levels versions. To motivate the choice of the correlation kernel, recall
that for $z=(x,\xi )\in \mathbb{R}^{2d}$, the short-time Fourier transform
of a function $f$ with respect to a window function $g\in L^{2}(\mathbb{R}%
^{d})$ is defined as \cite{Charly}: 
\begin{equation}
\mathcal{V}_{g}f(x,\xi )=\int_{\mathbb{R}^{d}}f(t)\overline{g(t-x)}e^{-2\pi
i\xi t}dt\text{.}  \label{Gabor}
\end{equation}%
For $d=1$ and $g(t)=h_{0}(t)=2^{1/4}e^{-\pi t^{2}}$, then, writing $z=x+i\xi 
$, then $\mathcal{V}_{h_{0}}f(x,-\xi )=e^{-i\pi x\xi }e^{-\frac{\pi }{2}%
\left\vert z\right\vert ^{2}}Bf(z)$\ where $Bf(z)$ is the Bargmann-Fock
transform%
\begin{equation*}
Bf(z)=2^{\frac{1}{4}}\int_{\mathbb{R}}f(t)e^{2\pi tz-\pi t^{2}-\frac{\pi }{2}%
z^{2}}dt\text{,}
\end{equation*}%
which maps $L^{2}(\mathbb{R})$ onto the Fock space of entire functions,
whose reproducing kernel is the kernel of the infinite Ginibre ensemble and
which, as a ressult, can be seen as a weighted version of $\mathcal{V}%
_{h_{0}}\left( L^{2}(\mathbb{R})\right) $. For\ choices of $g$ within the
family of Hermite functions $h_{n}(t)$, defined as in (\ref{Hermite}), one
obtains a sequence of transforms defined by $\mathcal{V}_{h_{n}}f(x,-\xi
)=e^{-i\pi x\xi }e^{-\frac{\pi }{2}\left\vert z\right\vert ^{2}}B^{(n)}(z)$,
and mapping $L^{2}(\mathbb{R})$ onto the eigenspaces of the Landau levels
operator, which are weighted versions of $\mathcal{V}_{h_{n}}\left( L^{2}(%
\mathbb{R})\right) $ \cite{Abre,APRT,abgrro17}.

\emph{The Weyl-Heisenberg ensemble}, introduced in\emph{\ }\cite{APRT} and
studied further in \cite{abgrro17,Partial}, is the family of DPPs $\mathcal{X%
}_{g}$ on $\mathbb{R}^{2d}$, with correlation kernel equal to the
reproducing kernel of $\mathcal{V}_{g}L^{2}(\mathbb{R}^{d})$: 
\begin{equation}
K_{g}(z,w)={K}_{g}((x,\xi ),(x^{\prime },\xi ^{\prime }))=\int_{\mathbb{R}%
^{d}}e^{2\pi i(\xi ^{\prime }-\xi )t}g(t-x^{\prime })\overline{g(t-x)}dt%
\text{,}  \label{eq:l1}
\end{equation}%
for some non-zero function $g\in L^{2}(\mathbb{R}^{d})$ with $\left\Vert
g\right\Vert _{L^{2}(\mathbb{R}^{d})}=1$\ and $(x,\xi ),(x^{\prime },\xi
^{\prime })\in \mathbb{R}^{2d}$. For $g$ a Hermite function, Weyl-Heisenberg
ensembles lead to the Ginibre type ensembles for higher Landau levels \cite%
{Shirai,APRT} (see the remark below) and to the Heisenberg family of DPPs 
\cite{Heisenberg}. The complex Ginibre ensemble as the prototypical
Weyl-Heisenberg ensemble follows by setting $d=1$ and choosing $g$ in~%
\eqref{eq:l1} to be the Gaussian $h_{0}(t)=2^{1/4}e^{-\pi t^{2}}$. The
resulting kernel is%
\begin{equation*}
{K}_{h_{0}}(z,w)=e^{i\pi (x^{\prime }\xi ^{\prime }-x\xi )}e^{-\frac{\pi }{2}%
(\left\vert z\right\vert ^{2}+\left\vert w\right\vert ^{2})}e^{\pi \overline{%
z}w},\qquad z=x+i\xi ,\,w=x^{\prime }+i\xi ^{\prime }.
\end{equation*}%
Modulo a phase factor, this is the kernel of the \emph{infinite Ginibre
ensemble} $K_{0}(z,w)=e^{-\frac{\pi }{2}(\left\vert z\right\vert
^{2}+\left\vert w\right\vert ^{2})}e^{\pi \overline{z}w}$. Choosing $%
h_{n}(t) $ a Hermite function, a similar relation holds between ${K}%
_{h_{n}}(z,w)$ and $K_{n}(z,w)$.

The area law is obtained for Berezin-Toeplitz operators on compact Kaehler
manifolds and for the Bargmann transform (including thus the first Landau
level case of the Ginibre DPP) in \cite{EntaglementEntropy}, but the
relation with the variance is not made explicit. In \cite{Ginibre}, a
proportionality relation between the entanglement entropy and the number
variance has been obtained for the \emph{finite} Ginibre ensemble (it is
unclear at the moment if the methods in this note can handle finite DPPs
since, in such cases, the higher order traces may be difficult to control).
For a discussion of the relations between entanglement entropy and variance
fluctuations in a broad sense, see \cite{NC}.

\begin{theorem}
Let $\Omega \subset \mathbb{R}^{2d}$\ compact.\textbf{\ }Let $K_{g}(z,w)$ be
the kernel of a Weyl-Heisenberg ensemble $\mathcal{X}_{g}$ with $g$
satisfying, for some $s\geq 1/2$, 
\begin{equation}
C_{g}=\left[ \int_{\mathbb{R}^{2d}}\left\vert V_{g}g(z)\right\vert dz\right]
^{2}\int_{\mathbb{R}^{2d}}(1+|z|)^{2s}\left\vert V_{g}g(z)\right\vert
^{2}dz<\infty \text{.}  \label{c}
\end{equation}%
Then the entanglement entropy of the Weyl-Heisenberg ensemble on $\Omega $
satisfies the area law%
\begin{equation*}
S\left( \mathcal{X}_{g}\mathcal{(}\Omega \mathcal{)}\right) \lesssim
\left\vert \partial {\Omega }\right\vert \text{.}
\end{equation*}
\end{theorem}

\begin{proof}
We follow \cite[(2.6)]{SD} and consider the Schatten quasinorm of the Hankel
operator such that $H^{\ast }H=T_{\Omega }-T_{\Omega }^{2}$. Then%
\begin{equation*}
\left\Vert H\right\Vert _{\widetilde{p}}^{\widetilde{p}}=\mathrm{trace}%
\left( \left( H^{\ast }H\right) ^{\frac{1}{2}}\right) ^{\widetilde{p}}=%
\mathrm{trace}\left( T_{\Omega }^{\widetilde{p}}\left( 1-T_{\Omega }\right)
^{\widetilde{p}}\right) ^{\frac{1}{2}}=\mathrm{trace}\left( h_{\widetilde{p}%
/2}\left( T_{\Omega }\right) \right) \text{,}
\end{equation*}%
where $h_{p}(x)=x^{p}(1-x)^{p}$.\ Thus, the results for $\widetilde{p}<2$ in
Proposition 3.1 in \cite{SD}, assuming (\ref{c}), assure, writing $p=\frac{1%
}{2}\widetilde{p}$ that, for $0<p<1$, $\mathrm{trace}\left( h_{\widetilde{p}%
/2}\left( T_{\Omega }\right) \right) =\mathrm{trace}\left( h_{p}\left(
T_{\Omega }\right) \right) =\mathrm{trace}\left( T_{\Omega }^{p}\left(
1-T_{\Omega }\right) ^{p}\right) $ is bounded by $C\left\vert \partial {%
\Omega }\right\vert >0$. Thus, we can apply Proposition 2.2 to yield%
\begin{equation*}
S\left( \mathcal{X}_{g}\mathcal{(}\Omega \mathcal{)}\right) =\mathrm{trace}%
(f(T_{\Omega }))\lesssim \mathbb{V}\left( \mathcal{X}_{g}(\Omega )\right) =%
\mathrm{trace}\left( T_{\Omega }-T_{\Omega }^{2}\right) \text{.}
\end{equation*}%
Let $\varphi \in {L^{1}({\mathbb{R}^{d}})}$ with $\int \varphi =1$. Then,
for a set $\Omega $ of finite perimeter $\left\vert \partial {\Omega }%
\right\vert $, Lemma 3.2 in \cite{AGR} gives 
\begin{equation*}
\lVert 1_{\Omega }\ast \varphi -1_{\Omega }\rVert _{L^{1}({\mathbb{R}^{2d}}%
)}\leq \left\vert \partial {\Omega }\right\vert \int_{\mathbb{R}%
^{2d}}\left\vert z\right\vert \left\vert \varphi (z)\right\vert dz\text{.}
\end{equation*}%
Applying this inequality with $\varphi (z)=\left\vert V_{g}g(z)\right\vert
^{2}$ and observing that $K_{g}(z,w)=V_{g}g(z-w)$ leads to%
\begin{eqnarray*}
\mathrm{trace}\left( T_{\Omega }-T_{\Omega }^{2}\right) &=&\left\vert
\int_{\Omega }\int_{\Omega }\varphi (z-w)dzdw-\int_{\Omega }dz\right\vert \\
&=&\left\vert \int_{\Omega }\left( 1_{\Omega }\ast \varphi (w)-1_{\Omega
}\right) dw\right\vert \leq \lVert 1_{\Omega }\ast \varphi -1_{\Omega
}\rVert _{L^{1}({\mathbb{R}^{2d}})}\leq C\left\vert \partial {\Omega }%
\right\vert \text{,}
\end{eqnarray*}%
where $C=\int_{\mathbb{R}^{2d}}\left\vert z\right\vert \left\vert
V_{g}g(z)\right\vert ^{2}dz$. This last bound has been obtained in a
different form in \cite{DeMarie}. The more direct proof presented is
implicit in \cite{AGR}.
\end{proof}

\begin{example}
The Landau operator acting on the Hilbert space $L^{2}\left( \mathbb{C},e^{-%
\frac{\pi }{2}\left\vert z\right\vert ^{2})}\right) $ can be defined as 
\begin{equation}
L_{z}:=-\partial _{z}\partial _{\overline{z}}+\pi \overline{z}\partial _{%
\overline{z}}\text{.}  \label{2.1.3}
\end{equation}%
The spectrum of $L_{z}$ is given by $\sigma (L_{z})=\{\pi n:n=0,1,2,\ldots
\} $. The eigenspaces have associated reproducing kernel \cite{AAZ} 
\begin{equation*}
K_{n}(z,w)=L_{n}(\pi \left\vert z-w\right\vert ^{2})e^{\pi \overline{z}w}%
\text{,}
\end{equation*}%
where $L_{n}$\ is a Laguerre polynomial. Let the window $g$ of the
Weyl-Heisenberg kernel be a Hermite function 
\begin{equation}
h_{n}(t)=\frac{2^{1/4}}{\sqrt{n!}}\left( \frac{-1}{2\sqrt{\pi }}\right)
^{n}e^{\pi t^{2}}\frac{d^{n}}{dt^{n}}\left( e^{-2\pi t^{2}}\right) ,\qquad
n\geq 0.  \label{Hermite}
\end{equation}%
Then 
\begin{equation*}
{K}_{h_{n}}(z,w)=e^{i\pi (x^{\prime }\xi ^{\prime }-x\xi )}e^{-\frac{\pi }{2}%
(\left\vert z\right\vert ^{2}+\left\vert w\right\vert ^{2})}L_{n}(\pi
\left\vert z-w\right\vert ^{2})e^{\pi \overline{z}w},\qquad z=x+i\xi
,\,w=x^{\prime }+i\xi ^{\prime }\text{.}
\end{equation*}%
Now, from Theorem 2.2, denoting by $\mathcal{X}_{n}$ the DPP associated to
the $nth$ Landau level, 
\begin{equation*}
S(\mathcal{X}_{n}(\Omega ))\lesssim \left\vert \partial {\Omega }\right\vert 
\text{.}
\end{equation*}%
Moreover, from \cite[Theorem 1.1]{Shirai} (see also \cite[page 3]%
{HyperbolicDPP} for an alternative proof), one has $S(\mathcal{X}%
_{n}(D_{r}))\sim C_{n}r$ as $r\rightarrow \infty $. It follows that%
\begin{equation*}
S(\mathcal{X}_{n}(D_{r}))\sim Cr\text{,}
\end{equation*}%
as $r\rightarrow \infty $, for some constant $C$.\ This is an area law (in $%
\mathbb{R}^{2}$) for the entanglement entropy of integer quantum Hall states
modelled by DPP on higher Landau levels (see also the limit case $\beta =1$
in Theorem 2.5 of \cite{FL}).
\end{example}

\begin{remark}
Now, putting together Proposition 4.2 and Lemma 4.3 of \cite{DeMarie}, and (%
\ref{lower}) we realize, that, under certain conditions on $g$, we have $%
\left\vert \partial {\Omega }\right\vert \lesssim \mathbb{V}\left( \mathcal{X%
}_{g}(\Omega )\right) $. Thus, under the conditions of Proposition 4.2 and
Lemma 4.3 in \cite{DeMarie}, we have a double bound for the growth of the
entanglement entropy of the Weyl-Heisenberg ensemble on $\Omega $:%
\begin{equation*}
\left\vert \partial {\Omega }\right\vert \lesssim S\left( \mathcal{X}_{g}%
\mathcal{(}\Omega \mathcal{)}\right) \lesssim \left\vert \partial {\Omega }%
\right\vert \text{.}
\end{equation*}
\end{remark}

The condition (\ref{Shatten}) has been verified in \cite{EntaglementEntropy}
under the assumption of Gaussian decay of the kernel, and the analysis
includes fermionic states on a K\"{a}hler manifold and the infinite Ginibre
ensemble. For the kernel corresponding to the Weyl-Heisenberg ensemble, the
first bounds were obtained in \cite{DeMarie} and the moderate decay (\ref{c}%
) required considerable technical work \cite{SD}.

\begin{remark}
For general $d$, class II hyperuniformity is characterized by the following
asymptotic growth of the variance on a compact region $\Omega \subset 
\mathbb{R}^{d}$ dilated by $L>0$%
\begin{equation*}
\mathbb{V}\left( \mathcal{X}(L\Omega )\right) \sim C_{\Omega }L^{d-1}\log L%
\text{, \ \ \ \ }L\rightarrow \infty \text{.}
\end{equation*}%
Thus, just using inequality (\ref{lower}), (which holds without assumptions
on the kernel, since it follows from an inequality valid pointwise), we
conclude at once the following: \emph{for a DPP }$X$\emph{\ in the Class II
hyperuniformity},%
\begin{equation*}
S\left( \mathcal{X}(L\Omega )\right) \geq O(L^{d-1}\log L)\text{,}
\end{equation*}%
as $L\rightarrow \infty $ leading to the \emph{violation of the area law},
due to the $\log L$ correction. Thus, every \emph{Class II hyperuniform DPP
violates the area law.}
\end{remark}

\section{Entanglement entropy and finite Weyl-Heisenberg ensembles}

A feature of the Weyl-Heisenberg ensemble is the possibility of constructing
finite-dimensional DPPs with first point intensity converging to the
indicator domain of a pre-defined compact region $\Omega $. Details of such
finite dimensional constructions are given in \cite{abgrro17}, where it is
shown, in the Hermite window case, that the resulting processes are closely
related to the finite polyanalytic Ginibre ensembles of \cite{HH}. We will
now sketch the construction of finite Weyl-Heisenberg ensembles. Since $%
\{e_{n}^{\Omega }(z):n\geq 1\}$ spans the space with reproducing kernel $%
K_{g}(z,w)$, we have%
\begin{equation*}
K_{g}(z,w)=\sum_{n\geq 1}e_{n}^{\Omega }(z)\overline{e_{n}^{\Omega }(w)}%
\text{.}
\end{equation*}

Now we define the finite Weyl-Heisenberg ensemble as follows (see the
introduction of \cite{abgrro17} for details).

\begin{definition}
Let $N_{\Omega }=\left\lfloor \Omega \right\rfloor $\ be the smallest
integer greater than or equal to $\left\vert \Omega \right\vert $.\ \emph{%
The finite Weyl-Heisenberg ensemble }$\mathcal{X}_{g}^{N_{\Omega }}$\emph{\ }%
is the determinantal point process (DPP) associated with the truncated kernel%
\begin{equation*}
K_{g}^{N_{\Omega }}(z,w)=\sum_{n=1}^{N_{\Omega }}e_{n}^{\Omega }(z)\overline{%
e_{n}^{\Omega }(w)}\text{.}
\end{equation*}
\end{definition}

\begin{example}
Consider the Gaussian $h_{0}(t)=2^{1/4}e^{-\pi t^{2}}$ leading to the
infinite Ginibre ensemble kernel%
\begin{equation}
K_{h_{0}}(z,w)=e^{i\pi (x^{\prime }\xi ^{\prime }-x\xi )}e^{-\frac{\pi }{2}%
(\left\vert z\right\vert ^{2}+\left\vert w\right\vert ^{2})}e^{\pi z%
\overline{w}}\text{.}  \label{Ginf}
\end{equation}%
Denote by $\left\vert D_{R}\right\vert =\pi R^{2}$ the area of the disc. The
eigenfunctions of 
\begin{equation*}
(T_{D_{R}}f)(z)=\int_{D_{R}}f(w)\overline{K_{h_{0}}(z,w)}dw\text{,}
\end{equation*}%
are $e_{n+1}^{N_{D_{R}}}(z)=\left( \frac{\pi ^{j}}{j!}\right) ^{\frac{1}{2}%
}e^{-i\pi x\xi }e^{-\frac{\pi }{2}\left\vert z\right\vert ^{2}}z^{n}$. The
corresponding kernel of the finite Weyl-Heisenberg ensemble on $D_{R}$ is
then 
\begin{equation}
K_{h_{0}}^{N_{D_{R}}}(z,w)=e^{i\pi (x^{\prime }\xi ^{\prime }-x\xi )}e^{-%
\frac{\pi }{2}(\left\vert z\right\vert ^{2}+\left\vert w\right\vert
^{2})}\sum_{n=0}^{N_{D_{R}}-1}\frac{\left( \pi z\overline{w}\right) ^{n}}{n!}%
\text{,}  \label{Gfin}
\end{equation}%
where $N_{D_{R}}=\left\lfloor \pi R^{2}\right\rfloor $. This is, modulo a
phase factor, the kernel of the finite Ginibre ensemble, obtained by
truncating the expansion of the exponential $e^{\pi z\overline{w}}$.
\end{example}

We now provide a bound on $S(\mathcal{X}_{g}(\Omega ))$ involving the number
of points of $\mathcal{X}_{g}^{N_{\Omega }}$ that in average fall in $\Omega 
$\ and which can be explicitly computed in terms of the first eigenvalues of 
$T_{\Omega }$.

\begin{theorem}
Let $\Omega \subset \mathbb{R}^{2d}$\ compact and $g$ satisfying (\ref{c}).
The entanglement entropy of the Weyl-Heisenberg ensemble on $\Omega $
satisfies%
\begin{equation*}
S\left( \mathcal{X}_{g}\mathcal{(}\Omega \mathcal{)}\right) \lesssim
N_{\Omega }-\mathbb{E}\left( \mathcal{X}_{g}^{N_{\Omega }}(\Omega )\right)
\end{equation*}%
or%
\begin{equation*}
S\left( \mathcal{X}_{g}\mathcal{(}\Omega \mathcal{)}\right) \lesssim
N_{\Omega }{-}\sum_{n=1}^{N_{\Omega }}\lambda _{n}^{\Omega }\text{.}
\end{equation*}
\end{theorem}

\begin{proof}
The $1$-point intensity of $\mathcal{X}_{g}^{N_{\Omega }}$ is%
\begin{equation*}
\rho _{1}^{N_{\Omega }}(z)=K_{g}^{N_{\Omega }}\left( z,z\right)
=\sum_{n=1}^{N_{\Omega }}\left\vert e_{n}^{\Omega }(z)\right\vert ^{2}\text{.%
}
\end{equation*}%
Thus,%
\begin{eqnarray}
\mathbb{E}\left( \mathcal{X}_{g}^{N_{\Omega }}(\Omega )\right)
&=&\int_{\Omega }K_{g}^{N_{\Omega }}(z,z)dz  \notag \\
&=&\sum_{n=1}^{N_{\Omega }}\int_{\Omega }\left\vert e_{n}^{\Omega
}(z)\right\vert ^{2}\,dz=\sum_{n=1}^{N_{\Omega }}{\lambda _{n}^{\Omega }}
\label{E_}
\end{eqnarray}%
and 
\begin{eqnarray}
\mathbb{V}\left( \mathcal{X}_{g}(\Omega )\right) &=&\sum_{n\geq 1}{\lambda
_{n}^{\Omega }}-\sum_{n\geq 1}{(\lambda _{n}^{\Omega }})^{2}=\sum_{n=1}^{N_{%
\Omega }}\lambda _{n}^{\Omega }(1-\lambda _{n}^{\Omega })+\sum_{n>N_{\Omega
}}\lambda _{n}^{\Omega }(1-\lambda _{n}^{\Omega })  \notag \\
&\leq &\sum_{n=1}^{N_{\Omega }}(1-\lambda _{n}^{\Omega })+\sum_{n>N_{\Omega
}}\lambda _{n}^{\Omega }  \label{ineq} \\
&=&N_{\Omega }-\sum_{n=1}^{N_{\Omega }}\lambda _{n}^{\Omega }+\mathrm{trace}%
(T_{\Omega })-\sum_{n=1}^{N_{\Omega }}\lambda _{n}^{\Omega }  \notag \\
&\leq &2N_{\Omega }-2\mathbb{E}\left( \mathcal{X}_{g}^{N_{\Omega }}(\Omega
)\right) \text{.}  \notag
\end{eqnarray}%
The result follows from the upper bound $S\left( \mathcal{X}_{g}\mathcal{(}%
\Omega \mathcal{)}\right) \lesssim \mathbb{V}\left( \mathcal{X}_{g}(\Omega
)\right) $.
\end{proof}

We finally bound the entanglement entropy of the Weyl-Heisenberg ensemble on 
$\Omega $ by the deviation of the $1$-point intensity $\rho _{1}^{N_{\Omega
}}(z)$ of the finite Weyl-Heisenberg ensemble $\mathcal{X}_{g}^{N_{\Omega }}$
from the flat density $1_{\Omega }$.

\begin{theorem}
Let $\Omega \subset \mathbb{R}^{2d}$\ compact and $g$ satisfying (\ref{c}).
The entanglement entropy of the Weyl-Heisenberg ensemble on $\Omega $
satisfies%
\begin{equation*}
\int_{\mathbb{R}^{2d}}\left\vert \rho _{1}^{N_{\Omega }}(z)-1_{\Omega
}(z)\right\vert \,dz\lesssim S\left( \mathcal{X}_{g}\mathcal{(}\Omega 
\mathcal{)}\right) \lesssim \int_{\mathbb{R}^{2d}}\left\vert \rho
_{1}^{N_{\Omega }}(z)-1_{\Omega }(z)\right\vert \,dz\text{.}
\end{equation*}
\end{theorem}

\begin{proof}
We start with inequality (\ref{ineq}) and then proceed as in the proof of
Theorem 1.6 in \cite{APR}:%
\begin{eqnarray*}
\mathbb{V}\left( \mathcal{X(}\Omega \mathcal{)}\right) &\leq
&\sum_{n=1}^{N_{\Omega }}(1-\lambda _{n}^{\Omega })+\sum_{n>N_{\Omega
}}\lambda _{n}^{\Omega } \\
&=&\int_{\mathbb{R}^{2d}-\Omega }\rho _{1}^{N_{\Omega }}(z)\,dz+\left( 
\mathrm{trace}(T_{\Omega })-\int_{\Omega }\rho _{1}^{N_{\Omega
}}(z)\,dz\right) \\
&=&\int_{\mathbb{R}^{2d}-\Omega }\left\vert \rho _{1}^{N_{\Omega
}}(z)-1_{\Omega }(z)\right\vert \,dz+\int_{\Omega }\left\vert \rho
_{1}^{N_{\Omega }}(z)-1_{\Omega }(z)\right\vert \,dz \\
&=&\int_{\mathbb{R}^{2d}}\left\vert \rho _{1}^{N_{\Omega }}(z)-1_{\Omega
}(z)\right\vert \,dz\text{.}
\end{eqnarray*}%
The result follows from (\ref{ineEE}). The lower bound of the variance
follows from a related argument, which is contained in the Steps 2 and 3 of
the proof of Theorem 1.5 in \cite{APR}.
\end{proof}

\begin{remark}
To obtain the previous theorem, we have proved that%
\begin{equation*}
\int_{\mathbb{R}^{2d}}\left\vert \rho _{1}^{N_{\Omega }}(z)-1_{\Omega
}(z)\right\vert \,dz\lesssim \mathbb{V}\left( \mathcal{X(}\Omega \mathcal{)}%
\right) \lesssim \int_{\mathbb{R}^{2d}}\left\vert \rho _{1}^{N_{\Omega
}}(z)-1_{\Omega }(z)\right\vert \,dz\text{.}
\end{equation*}%
This holds for a DPP with no restrictions (details can be provided for a
general case, but this would be out of scope of this note). Thus, all
conditions for hyperuniformity of DPPs can be written using, instead of the
variance $\mathbb{V}\left( \mathcal{X(}\Omega \mathcal{)}\right) $ of $%
\mathcal{X}$, the $L^{1}$ rate of convergence of the associated finite DPP $%
\mathcal{X}^{N_{\Omega }}$, $\int_{\mathbb{R}^{2d}}\left\vert \rho
_{1}^{N_{\Omega }}(z)-1_{\Omega }(z)\right\vert \,dz$.
\end{remark}

\textbf{Acknowledgement.} \emph{I would like to thank the three reviewers
for comments that helped putting this work in the proper background context,
and for detecting an innacuracy in the previous formulation and proof of
Proposition 2.2.}

\end{document}